\documentclass[conference]{IEEEtran}
\IEEEoverridecommandlockouts

\usepackage{cite}
\usepackage{amsmath,amssymb,amsfonts}
\usepackage{algorithmic}
\usepackage{graphicx}
\usepackage{textcomp}
\usepackage{xcolor}
\def\BibTeX{{\rm B\kern-.05em{\sc i\kern-.025em b}\kern-.08em
    T\kern-.1667em\lower.7ex\hbox{E}\kern-.125emX}}

\usepackage{graphicx}
\usepackage{amsthm}
\usepackage{algorithm}
\usepackage{algorithmic}

\usepackage{mathtools}
\usepackage{xcolor}
\usepackage{listings}

\usepackage{xcolor}
\usepackage{listings}

\usepackage{tikz}
\usepackage{tikz-cd}
\usepackage{float}
\usepackage{color}

\newtheorem{theorem}{Theorem}

\newtheorem{definition}{Definition}[section]
\newtheorem{proposition}{Proposition}

\newcommand{\ket}[1]{\left \rvert #1 \right \rangle}

\newcommand{\remove}[1]{}
\newcommand{\mg}{\color{magenta}}

\newcommand{\SD}{{\sf D-CD}}
\newcommand{\PCD}{{\sf  pHE-CD}}

\begin{document}

\title{Hybrid Encryption with Certified Deletion in Preprocessing Model}

\author{%
  \IEEEauthorblockN{Kunal Dey}
  \IEEEauthorblockA{Department of Computer Science and Engineering\\
                    SRM University-AP\\
                    Andhra Pradesh, India\\
                    }
  \and
  \IEEEauthorblockN{Reihaneh Safavi-Naini}
\IEEEauthorblockA{Department of Computer Science\\
                    University of Calgary\\
                    Calgary, Alberta, Canada\\
                    }
}


\maketitle

\begin{abstract}
Encryption with certified deletion (CD) allows Alice to outsource data to Bob in encrypted form, and later either provide Bob with the decryption key allowing him to recover the data, or request deletion of the ciphertext, in which case,
Bob  can verifiably delete the ciphertext. In the latter case Boob can produce a certificate that, if  valid, ensures that even revealing the key does not leak any information about the plaintext to Bob.
The functionality, while impossible to realize using classical information alone can be achieved using quantum information. Encryption with CD  has been proposed both in the symmetric-key setting, using a one-time pad (OTP), and in the public-key setting (PKE-CD), where Bob has a key pair that can be reused.
PKE-CD  
uses computational hardness assumptions that are 
vulnerable to 
advances in 
computing.

In this work, we introduce and formalize {\em hybrid encryption with certified deletion in the preprocessing} model (pHE-CD) that avoids 
hardness assumptions, and propose two constructions. The constructions compose an information-theoretic key encapsulation mechanism (iKEM) with a data encapsulation mechanism (DEM) that provides certified deletion (DEM-CD) security, offering different security guarantees
 depending on the security properties of the DEM-CD. 
When DEM-CD is one-time information theoretically secure, the composition provides {\em information-theoretic security} for both encryption and certified deletion. When DEM-CD is computationally secure, 
the 
pHE-CD provides computational security 
for encryption 
and {\em everlasting security} for certified deletion, which 
means that confidentiality is computational prior to  verification of   deletion certificate, 
and becomes information theoretic, after  successful verification of 
certificate, guaranteeing that the data has been removed from the 
adversary's view in an information-theoretic sense. The  construction  is key efficient and uses a DEM-CD that is constructed using quantum coding and AES, providing quantum-resistant security for encryption without using any hardness assumption. 
%
%
Both pHE-CD schemes are for encryption of arbitrarily long messages.  We discuss our results and directions for future work. 
\end{abstract}


\section{Introduction}
Consider the scenario that Alice is outsourcing an encrypted file to Bob, and at a later time requires Bob to delete the file. Achieving a guarantee for Alice that Bob has actually deleted the file is impossible using only classical information, because digital data can be copied perfectly, making it vulnerable to future key leakage or to attacks on the encryption scheme.  By leveraging quantum information, in particular the no-cloning theorem~\cite{dieks1982communication, wootters1982single, park1970concept},  
Broadbent and Islam (BI)~\cite{broadbent2020quantum} 
proposed an elegant solution to this problem, known as encryption with {\em certified deletion} \cite{unruh2015revocable}. BI's scheme is a symmetric-key encryption scheme with \emph{information-theoretic security} that allows Bob to make one of two mutually exclusive choices: either (i) delete the ciphertext and produce a \emph{certificate of deletion} which, when verified by Alice, guarantees that the file is information-theoretically removed from Bob’s view, or (ii) decrypt the ciphertext and learn the plaintext, but \emph{not both}.
The scheme is particularly attractive because it employs \emph{Wiesner coding} \cite{wiesner1983conjugate}, a quantum encoding technique closely related to the BB84 quantum key distribution (QKD) protocol of Bennett and Brassard~\cite{bennett2014quantum}, making it implementable using today’s quantum technology \cite{tomamichel2012tight,bacco2013experimental, diamanti2016practical, hayashi2006practical}. 
The BI scheme, however, is one-time and, similar to the OTP~\cite{
shannon1949mathematical}, 
requires a fresh key to be  generated by Alice and sent to Bob 
for each encryption, with secret-key length 
at least the message length, plus additional overhead required for the quantum encoding and deletion verification.

\vspace{-0.8mm}
 Hiroka et al.\ proposed \emph{public-key encryption with certified deletion (PKE-CD)}~\cite{hiroka2021quantum}, for a cloud-centric model that allows an encryptor, say Alice, to encrypt a message for Bob using Bob's public key. The ciphertext is stored in the cloud, which can either delete it and produce a deletion certificate, or make it available to Bob for decryption, but not both. Alice can require the cloud to delete the ciphertext, and be assured that, if the deletion certificate is valid, no information about the message will be leaked to the cloud, even if Bob's private key subsequently becomes available to it. Alternatively, Alice can instruct the cloud to forward the ciphertext to Bob, in which case Bob can decrypt the message using his private key. (This can be done securely because of the no-cloning theorem \cite{dieks1982communication,wootters1982single,park1970concept}.) This separation of the ciphertext holder (the cloud) from the decryption key holder (Bob) is what gives certified deletion its power and allows Bob to {\em re-use} his key. Deletion security in~\cite{hiroka2021quantum} is against a computationally bounded adversary, whereas in~\cite{hiroka2024certified} is \emph{ certified everlasting}: prior to certificate validation the adversary is computationally bounded, but once a valid deletion certificate is issued, security holds against an unbounded adversary (that can have  access to decryption key also).

Bartusek et al.~\cite{bartusek2023cryptography} proposed a new approach to constructing PKE-CD with \emph{certified everlasting security} for one-bit messages using 2-out-of-2 quantum secret sharing with certified deletion. 
They also showed 
that the construction
can be extended to multi-bit messages by repeated application of the one-bit scheme.





All PKE-CD schemes use a post-quantum public-key encryption scheme 
for classical information that relies on computational hardness assumptions.

{\em Our work.} Our goal is to 
 replace the computational hardness assumptions in PKE-CD. 
We propose two constructions: one achieving 
information-theoretic deletion security, and one achieving 
\emph{everlasting} deletion security. Both constructions
rely on fundamental principles of quantum information theory. The 
second construction, additionally relies on the security of symmetric-key 
primitives, such as AES, against quantum polynomial-time (QPT) adversaries. 


\vspace{2pt}
{\em Our approach and outline of the schemes.}
We use {\em Hybrid Encryption in the correlated  randomness model}, also referred to as {\em Hybrid Encryption in pre-processing  model (pHE)} \cite{sharifian2021information}.
 pHE extends 
 the KEM/DEM hybrid encryption  system, that was defined in the public-key 
setting (HPKE) and was formalized in~\cite{cramer1998practical,Herranz2006KEMDEMNA} 
for efficient encryption of arbitrarily long messages, to the {\em correlated 
randomness setting}, where Alice and Bob have access to private correlated strings.
pHE  
has been defined 
with security against either a computationally unbounded or a computationally bounded adversary. We focus on 
\emph{information-theoretic KEM (iKEM)} which uses correlated randomness~\cite{maurer1993secret}  between Alice and Bob to establish a
shared key with information-theoretic security between them \cite{sharifian2021information}.
The DEM component in pHE can be implemented using an OTP
, resulting in a pHE 
with information-theoretic security,  or a {\em  computationally secure one-time symmetric key encryption system } 
that can be constructed using a block cipher algorithm such as AES, resulting in a  quantum-resistant  pHE. The composition theorem of pHE \cite[Theorem 2]{sharifian2021information} guarantees security of the composed systems. 

We adopt the pHE framework and compose an
iKEM 
with  {\em  DEMs
that support certified deletion (DEM-CD),}  thereby obtaining \emph{hybrid encryption scheme}s  with certified deletion in the preprocessing model (pHE-CD).  We define appropriate notions of encryption and deletion security for 
pHE-CD against both computationally unbounded and computationally bounded 
adversaries, and prove composition theorems for the security of the two 
constructions.  See Section~\ref{Sec:pHE-CD}.

\vspace{1pt}
\noindent


\section{Preliminaries}

A function $f(n)$ is called negligible, denoted $negl(n)$, if it 
approaches zero faster than the inverse of any polynomial in $n$. That 
is, for every constant $c > 0$, there exists $N_c$ such that for all 
$n > N_c$, we have
$f(n) < \frac{1}{n^c}$.

Let $z \leftarrow Z$ denote 
sampling an element $z$ uniformly at random from a finite set $Z$, and let $z \gets A(x)$ denote 
that $z$ is assigned the output of the 
algorithm $A$ (quantum, probabilistic, or deterministic) on input $x$.

\noindent {\em Quantum background.} A qubit is represented as a unit vector in a two-dimensional Hilbert space $\mathcal{H}$. Two commonly used orthonormal basis for qubit representation are computational basis $\mathbf{Z} =\{|0\rangle, |1\rangle\}$, where $|0\rangle = (1,0)^T$ and $|1\rangle = (0,1)^T$, and Hadamard basis $\mathbf{X} =\{|+\rangle, |-\rangle\}$ with $|+\rangle = \tfrac{1}{\sqrt{2}}(|0\rangle + |1\rangle)$ and $|-\rangle = \tfrac{1}{\sqrt{2}}(|0\rangle - |1\rangle)$.
The $X$ basis is mutually unbiased with respect to the computational basis.

{\em Wiesner’s conjugate coding}, encodes a bit $x \in \{0,1\}$ 
in basis $\theta \in \{0,1\}$ as $|x\rangle_{\theta}$, where $\theta=0$ and $\theta=1$ corresponding to computational and Hadamard basis, respectively.
This gives 
 four 
possible 
 states: $|0\rangle_{0} = |0\rangle, |1\rangle_{0} = |+\rangle, |0\rangle_{1} = |1\rangle, |1\rangle_{1} = |-\rangle$. The measurement of a qubit $|\phi\rangle$ in basis $\theta \in \{\mathbf{Z}, \mathbf{X}\}$ returns outcome $x \in \{0,1\}$ with probability $|\langle x_{\theta} | \phi \rangle|^2$ and projects the state onto $|x_{\theta}\rangle$.

\noindent We give an overview of 
the required
cryptographic systems. 
\noindent \textbf{Hybrid Encryption in Preprocessing Model (pHE).} 
pHE \cite{sharifian2021information} 
 considers a setting  
where Alice and Bob have  samples of  correlated random variables that are partially leaked to Eve.
An 
  iKEM is 
defined  
analogously
 to a public-key KEM, but with security 
against a computationally unbounded adversary.

\begin{definition}[Information-Theoretic KEM (iKEM) \cite{sharifian2021information}]\label{Def:iKEM}
 \vspace{-1pt}
For a security parameter $\lambda$,  
 a public distribution $P(X,Y,Z)$ and a key space $\{0,1\}^{iK.Len(\lambda)}$, 
an   iKEM {\sf iK} is defned by a tuple of algorithms {\sf (iK.Gen, iK.Enc, iK.Dec)} where 
${\sf iK.Gen}(1^\lambda, P)$ generates private correlated samples $(X, Y, Z)$ for Alice, Bob, and Eve, respectively; ${\sf iK.Enc}(X)$ takes $X$ and outputs 
a uniformly random key $K \in \{0,1\}^{iK.Len(\lambda)}$ and  a ciphertext $C$; and 
${\sf iK.Dec}(Y, C)$  that takes $Y$ and $C$, and outputs  $K$ or $\perp$.
%
\vspace{-3pt}
\end{definition}

Security of iKEM is defined against a {\em computationally unbounded adversary} that may query the {\em encapsulation} and {\em decapsulation oracles}, where each oracle implements the corresponding algorithm using the respective private variables 
$X $and 
$Y$. 
Security is defined against  three types of  adversaries: {\em OT (One-Time) adversary}, {\em $q_e$-CEA (Chosen Encapsulation Attack) Adversary}, and {\em $(q_e, q_d)$-CCA (Chosen Ciphertext Attack) Adversary} \cite[Definition 5]{sharifian2021information}.
\vspace{-2pt}
\begin{definition}[Data Encapsulation Mechanism (DEM)]\label{def:dem}
A DEM with security parameter $\lambda$, key space $\mathcal{K} = \{0,1\}^{Dem.Len(\lambda)}$, and an unrestricted message space is defined in terms of 
three algorithms: 
${\sf DEM.Gen}(1^\lambda)$ generates a key $K \in \mathcal{K}$; 
${\sf DEM.Enc}(M, K)$ encrypts a message $M$ using the key $K$ and produces a ciphertext $C$; 
${\sf DEM.Dec}(C, K)$ decrypts $C$ using the same key and returns $M$ or $\perp$. The first two algorithms are randomized, and the last one is deterministic.
 \vspace{-3pt}
\end{definition}
\vspace{-1mm}
 We consider DEM 
with perfect correctness.  Security of the DEM is defined using computational indistinguishability notions such 
as IND-OT and IND-OTCCA, under attack models including CPA, CCA1, and 
CCA2 \cite[Definition 2.5]{herranz2006kem}. 

\vspace{1pt}
\noindent
{\em Hybrid Encryption in Preprocessing Model (pHE)} consists of an iKEM and a DEM with compatible key spaces. pHE securely encrypts (and decrypts) messages of unrestricted length assuming an initial public distribution $P(XYZ)$ and using the algorithms of  
its building blocks, an iKEM and a DEM.

Security of pHE   can be defined against a computationally unbounded, or bounded, adversaries  using DEMs with corresponding  notions of security (indistinguishability aganist OT  and $q_e$-CPA, $(q_e,q_d)$-CCA adversaries)   \cite{sharifian2021information,panja2025hybrid}.

 \vspace{-2mm}
\begin{theorem}[{\cite[Theorem 2]{sharifian2021information}}]
\label{Composable security}
   Let ${\sf iK}$ be $\sigma(\lambda)$-IND-$q_e$-CEA secure (information-theoretic) and ${\sf D}$ be $\sigma'(\lambda)$-IND-OT secure (computational). Then the resulting hybrid scheme ${\sf pHE}$ is $[\sigma(\lambda)+\sigma'(\lambda)]$-IND-$q_e$-CPA secure against any computationally bounded adversary.
\end{theorem}

\noindent \textbf{SKE with Certified Deletion.} We review 
symmetric-key encryption with 
certified deletion property, 
introduced in \cite{hiroka2024certified}.

\vspace{-3pt}
\begin{definition}[Syntax]
\label{Def:SKECD}
    A SKE with certified deletion {\sf SKE-CD = (SKE-CD.KeyGen, SKE-CD.Enc, SKE-CD.Dec, SKE-CD.Del, SKE-CD.Vrfy)} is a tuple of QPT algorithms with a security parameter $\lambda$, plaintext space $\mathcal{M}$ and key space $\mathcal{K}$.

      \noindent
    --{\sf SKE-CD.KeyGen}$(1^\lambda)\rightarrow K$: The (classical) algorithm generates a uniformly random (classical) key $K$ from 
    $\mathcal{K}$.\\
 --{\sf SKE-CD.Enc}$(K,m) \rightarrow (CT,vk)$: 
    The quantum algorithm generates a quantum ciphertext ${CT}$ and a verification key $vk$, given a classical message $m \in \mathcal{M}$ and a secret key $K$.\\
 --{\sf SKE-CD.Dec}$({CT},K) \rightarrow (m^\prime \text{or} \perp)$: The quantum algorithm outputs m$^\prime$ or $\perp$, given a quantum ciphertext ${CT}$ and the secret key $K$. \\
  --{\sf SKE-CD.Del}$({CT})\rightarrow cert$: The 
  algorithm outputs a {\em classical deletion certificate $cert$  } given a quantum ciphertext ${CT}$.\\
 --{\sf SKE-CD.Vrfy}$(cert, vk) \rightarrow (\top~\text{or}~\bot)$: The classical algorithm outputs $\top$ or $\bot$, given two classical strings $cert$ and $vk$.

\remove{\begin{enumerate}
    \item {\sf SKE-CD.KeyGen}$(1^\lambda)\rightarrow K$: The classical algorithm generates a classical key $K$ that follows a uniform distribution in $\mathcal{K}$.
    \item {\sf SKE-CD.Enc}$(K,m) \rightarrow (CT,vk)$: 
    The quantum algorithm generates a quantum ciphertext ${CT}$ and a verification key $vk$, given a classical message $m \in \mathcal{M}$ and a secret key $K$.
    \item {\sf SKE-CD.Dec}$({CT},K) \rightarrow (m^\prime \text{or} \perp)$: The quantum algorithm outputs m$^\prime$ or $\perp$, given a quantum ciphertext ${CT}$ and the secret key $K$. 
    \item {\sf SKE-CD.Del}$({CT})\rightarrow cert$: The quantum algorithm outputs a classical deletion certificate $cert$ given a quantum ciphertext ${CT}$.
    \item {\sf SKE-CD.Vrfy}$(cert, vk) \rightarrow (\top \text{or} \bot)$: The classical algorithm outputs $\top$ or $\bot$, given two classical strings $cert$ and $vk$.
   
     \end{enumerate}}
\end{definition}

\begin{definition}[Correctness]
\label{Def: SKECD correct}
Correctness is defined for decryption (standard definition of encryption correctness with perfect decryptability) and deletion verification  as below. 

%
  %
\noindent \textbf{Verification correctness}: For any $\lambda \in \mathbb{N}, m \in \mathcal{M}$, given:
\begin{align*}
(K) &\leftarrow \text{\sf SKE-CD.KeyGen}(1^\lambda),\\
(CT,vk) &\leftarrow \text{\sf SKE-CD.Enc}(K,m),\\
\mathit{cert} &\leftarrow \text{\sf SKE-CD.Del}(CT).
\vspace{-2mm}
\end{align*}
\vspace{-2mm}
we have 
$
\Pr\big[\text{\sf SKE-CD.Vrfy}(vk,\mathit{cert})=\top\big]=1.
$
\end{definition}

\noindent \textbf{Security:} 
We consider one-time encryption security of SKE-CD against both computationally unbounded
\cite[Definition 2.11]{broadbent2020quantum} 
and computationally bounded, adversaries \cite{hiroka2021quantum}.
%
These are in-line with the corresponding  
IND-OT security definitions 
for DEM \cite{herranz2006kem}. (The definitions can be extended to multi-time security \cite[Definition C.6]{hiroka2024certified}.)

\begin{definition}[One-time Everlasting Certified Deletion Security \cite{hiroka2024certified}]
\label{Def: SKE-CD EV-CD}
Let {\sf SKE-CD = ( SKE-CD.KeyGen,  SKE-CD.Enc,  SKE-CD.Dec,  SKE-CD.Del,  SKE-CD.Vrfy)} be    an SKE with certified deletion. 
We define the following experiment for \emph{ one-time IND-OT  
everlasting certified deletion} security,
denoted by $\text{Exp}_{\text{SKE-CD}, \mathcal{A}}^{\text{ev-cd}}(\lambda, b)$, for a two stage adversary $\mathcal{A} = (\mathcal{A}_1, \mathcal{A}_2)$ where $\mathcal{A}_1$ is computationally bounded and $\mathcal{A}_2$ is computationally unbounded.
\begin{enumerate}
    \item The challenger computes $K \leftarrow \text{\sf SKE-CD.KeyGen}(1^\lambda)$.
    \item $\mathcal{A}_1$ sends $(m_0, m_1)$ to the challenger.
    \item The challenger  does
    $({CT},vk) \leftarrow \text{\sf SKE-CD.Enc}(K, m_b)$ and sends ${CT}$ to $\mathcal{A}_1$.
    \item At some point, $\mathcal{A}_1$ sends $cert$ to the challenger and sends the internal state to $\mathcal{A}_2$.
    \item The challenger computes $\text{\sf SKE-CD.Vrfy}(cert,vk)$. If the output is $\bot$, the challenger sends $\bot$ to $\mathcal{A}_2$. If the output is $\top$, the challenger sends $K$ to $\mathcal{A}_2$.
    \item $\mathcal{A}_2$ outputs its guess $b' \in \{0, 1\}$. 
    
\end{enumerate}
The experiment outputs $b'$.
Let $\text{Adv}_{\text{SKE-CD}, \mathcal{A}}^{\text{ot-cd}}(\lambda)$ be the advantage {\mg of } the adversary $\mathcal{A}$ in the above experiment. We define {\sf SKE-CD} to be one-time EV-CD secure if, for any adversary $\mathcal{A}$, it holds that
{\small
\begin{align*}
    \text{Adv}_{\text{SKE-CD}, \mathcal{A}}^{\text{ev-cd}}(\lambda) 
    & \triangleq \bigl|\Pr[\text{Exp}_{\text{SKE-CD}, \mathcal{A}}^{\text{ev-cd}}(\lambda, 0) = 1] \\
    &\quad - \Pr[\text{Exp}_{\text{SKE-CD}, \mathcal{A}}^{\text{ev-cd}}(\lambda, 1) = 1]\bigr| 
    \leq \text{negl}(\lambda).
\end{align*}
}
\end{definition}

\section{Hybrid Encryption with Certified Deletion in Preprocessing Model (pHE-CD)}
\label{Sec:pHE-CD}

We define \emph{hybrid encryption with certified deletion in the preprocessing model} (pHE-CD) and 
its associated security notions. 
We  begin by introducing 
data encapsulation mechanism with certified deletion (DEM-CD).\\
\noindent  \textbf{DEM with Certified Deletion.} 
A data encapsulation mechanism with certified deletion (DEM-CD) follows the same syntax as the SKE-CD scheme defined in Definition~\ref{Def:SKECD}, while adopting the standard encryption security notion of DEM. In this work, we consider DEM-CD schemes where 
deletion security is defined using the notion of one-time everlasting certified deletion security for SKE-CD (Definition~\ref{Def: SKE-CD EV-CD}).
\vspace{-1mm}
\begin{definition}[pHE-CD] For a security parameter, $\lambda$, let {\sf iK} = ({\sf iK.KeyGen, iK.Encap, iK.Decap}) and {\sf \SD = (\SD.KeyGen, \SD.{Encap}, \SD.Decap, \SD.{Del}, \SD.Vrfy)} represent a iKEM and a DEM-CD scheme, respectively, with the same key space $\mathcal{K} = \{0,1\}^{\text{pHE.Len}(\lambda)}$. For a public  distribution $P(XYZ)$, a hybrid encryption with certified deletion scheme in Preprocessing Model, denoted by {\sf \PCD} = ({\sf \PCD.KeyGen}, {\sf \PCD.\text{Enc}}, {\sf \PCD.\text{Dec}}, {\sf \PCD.\text{Del}}, {\sf \PCD.Vrfy}) is given in Figure \ref{fig:pHE-CD}\footnote{Note that in the ciphertext $CT = (C_1, C_2)$, the iKEM ciphertext $C_1$ is classical, while the DEM-CD ciphertext $C_2$ is quantum. Deletion algorithm is performed solely on the quantum component.}, where the message space is an unrestricted message space $\{0, 1\}^*$ for all $\lambda$. 
\begin{figure}
\centering
\footnotesize
\setlength{\tabcolsep}{3pt} 
\begin{tabular}{@{}p{0.47\columnwidth} p{0.47\columnwidth}@{}}
$\mathbf{Alg}\ \text{\sf \PCD.KeyGen}(1^\lambda,P)$ 
& 
$\mathbf{Alg}\ \text{\sf \PCD.Enc}(X,m)$
\\ 
$(X,Y,Z)\stackrel{\$}\gets \text{\sf iK.KeyGen}(1^\lambda)$
&
$(K,C_1)\stackrel{\$}\gets \text{\sf iK.Encap}(X)$
\\
Return $(X,Y,Z)$
&
$(vk,C_2)\gets \text{\sf \SD.Encap}(m,K)$
\\
&
Return $(vk,CT=(C_1,C_2))$
\\
\\
$\mathbf{Alg}\ \text{\sf \PCD.Dec}(Y,CT)$
&
$\mathbf{Alg}\ \text{\sf \PCD.Del}(CT)$
\\ 
Parse $CT=(C_1,C_2)$
&
Parse $CT=(C_1,C_2)$
\\
$K\gets \text{\sf iK.Decap}(Y,C_1)$
&
$cert \gets \text{\sf \SD.Del}(C_2)$
\\
If $K=\perp$: Return $\perp$
&
Return $cert$
\\
Else $m\gets \text{\sf \SD.Decap}(K,C_2)$
&
\\
Return $m$
&
\\
\end{tabular}

\vspace{0.2cm}

\centering
{\begin{tabular}{l}
$\mathbf{Alg}\ \text{\sf \PCD.Vrfy}(vk,cert)$
\\ 
Return $(\top \text{ or } \bot)\gets \text{\sf \SD.Vrfy}(cert,vk)$
\end{tabular}}

\caption{Hybrid encryption with certified deletion in preprocessing model}
\label{fig:pHE-CD}

\end{figure}
\vspace{-3mm}
\end{definition}

\noindent \textbf{Security:} pHE-CD is required to satisfy two security notions: (i) encryption security and (ii) certified deletion security. For encryption, we adopt the standard IND-$q_e$-CPA security notion for pHE. For certified deletion security we define {\em Everlasting Certified Deletion Security of pHE-CD against $q_e$ encryption queries (EV-$q_e$-CD)}, which is formalized below.
\vspace{-2mm}
\begin{definition}[Everlasting Certified Deletion Security of  {\sf \PCD}: EV-$q_e$-CD]
\label{Def: pHE-CD sec}
 For a   security parameter $\lambda$,  let {\sf \PCD} = ({\sf \PCD.KeyGen}, {\sf \PCD.{Enc}}, {\sf \PCD.{Dec}}, {\sf \PCD.{Del}}, {\sf \PCD.Vrfy}) be a hybrid encryption {rd system} in preprocessing model that provides 
the certified deletion property, and is constructed using an iKEM 
{\sf iK} and a DEM-CD 
{\sf \SD}. We consider experiment $\text{Exp}_{\text{\PCD}, \mathcal{A}}^{\text{ev-}q_e \text{-cd}}(\lambda, b)$ for {\em everlasting certified deletion security, where the adversary has access to
 $q_e$ encryption queries}. The experiment defines a game which is played between a challenger and a two-stage adversary $\mathcal{A} = (\mathcal{A}_1,\mathcal{A}_2)$, where $\mathcal{A}_1$ is a QPT adversary and $\mathcal{A}_2$ is an unbounded quantum adversary. The game proceeds as follows.
\begin{enumerate}
    \item The challenger samples,\\ $(X,Y,Z) \leftarrow \text{\sf \PCD.KeyGen}(1^\lambda,P)$, and provides $Z$ to the stage 1 adversary, $\mathcal{A}_1$.
    \item $\mathcal{A}_1$ has access to $q_e$ queries to the encryption oracle $\text{\sf \PCD.{Enc}}(X)$. 
    \item $\mathcal{A}_1$ provides two messages $(m_0, m_1)$ to the challenger.
    \item The challenger computes $({vk}, {CT}) \leftarrow \text{\sf \PCD.Enc}(X, m_b)$ and sends ${CT}$ to $\mathcal{A}_1$. Note that ${CT} = (C_1, {C_2})$, $(K, C_1) \leftarrow {\sf iK.Encap}(X)$ and $(vk,C_2) \leftarrow \text{\sf \SD.Encap}(m_b,K)$.
    
    \item  In stage 2, 
    $\mathcal{A}_1$ provides the deletion certificate, ${cert}$, to the challenger, and passes its internal state to $\mathcal{A}_2$.
    \item The challenger verifies the certificate using, $\text{\sf \PCD.vrfy}({vk}, {cert})$. If the output is $\bot$, it sends $\bot$ to $\mathcal{A}_2$. If the output is $\top$, the challenger sends $K$ to $\mathcal{A}_2$.
    \item $\mathcal{A}_2$ outputs its guess bit, $b' \in \{0, 1\}$.
    \item The experiment outputs $b'$ if the challenger outputs $\top$; otherwise, it outputs $\bot$.
\end{enumerate}
Let $\text{Adv}_{\text{\PCD}, \mathcal{A}}^{\text{ev-}q_e \text{-cd}}(\lambda)$ be the advantage of the experiment above. We say that the {\em {\sf \PCD} is $\text{EV-}q_e \text{-CD}$ secure} if for any two-stage quantum adversary $\mathcal{A}$ in the above experiment, it holds that,
\[
\begin{split}
\text{Adv}_{\text{\PCD}, \mathcal{A}}^{\text{ev-}q_e \text{-cd}}(\lambda) 
&\triangleq \bigl|\Pr[\text{Exp}_{\text{\PCD}, \mathcal{A}}^{\text{ev-}q_e \text{-cd}}(\lambda, 0) = 1] \\
&\quad - \Pr[\text{Exp}_{\text{\PCD}, \mathcal{A}}^{\text{ev-}q_e \text{-cd}}(\lambda, 1) = 1]\bigr| \leq \text{negl}(\lambda).
\end{split}
\]
\end{definition}

We  
define  pHE-CD with information-theoretic certified deletion security 
by allowing 
$\mathcal{A}_1$ in the above definition to be a  computationally  unbounded quantum adversary;  all other parts of the definition remain unchanged. 
\vspace{-1mm}
\begin{theorem}[{Encryption Security}]
\label{Thm: pHECD encryption}
    
    If {\sf iK} is information-theoretically secure $\text{IND-}q_e\text{-CEA}$  and {\SD} is IND-OT secure, either computationally or information-theoretically, then {\PCD} is $\text{IND-}q_e\text{-CPA}$ secure, respectively computationally or information-theoretically.
\end{theorem}
\vspace{-2mm}
The encryption security of the composed pHE-CD, follows directly from Theorem~\ref{Composable security}, resulting in an $\text{IND-}q_e \text{-CPA}$ secure {\PCD}. The following theorem {\mg state} deletion security of the composition of an IND-$q_e$-CEA secure iKEM and a one-time EV-CD secure DEM-CD, resulting in an EV-$q_e$-CD secure pHE-CD. Security proof reduces security of pHE-CD to the security of its two components. 




\begin{theorem}[{Certified Deletion Security}]
\label{Thm: pHECD compose}
    If {\sf iKEM} is 
    $\text{IND-}q_e \text{-CEA}$ secure and {\SD} is one-time EV-CD secure (information theoretic or computational), then {\PCD} is $\text{EV-}q_e \text{-CD}$ secure (information theoretic or computational, respectively) the security.
\end{theorem}
\vspace{-2mm}
\noindent 
{\em Proof overview:} The proof follows a standard hybrid-game argument showing that any efficient adversary that breaks EV-$q_e$-CD security of  pHE-CD can be transformed into an efficient adversary that breaks either the iKEM or the DEM-CD security. We define two games: $G_0$, which is identical to the deletion security game defined by the EV-$q_e$-CD experiment (Definition \ref{Def: pHE-CD sec}) for pHE-CD, and $G_1$, which differs only in that the DEM-CD encryption uses an independently sampled key instead of the iKEM-encapsulated key. If an adversary can distinguish $G_0$ from $G_1$, then one can construct an adversary that breaks the IND-$q_e$-CEA security of the iKEM, bounding the difference between the success probabilities of the adversary in the two games by the advantage of the iKEM adversary in winning its IND-$q_e$-CEA game. We also note that in $G_1$, the DEM-CD key is independent of the iKEM ciphertext and so distinguishing encryptions of the two messages after verification of deletion certificate, yields an adversary that breaks the one-time EV-CD security of DEM-CD, leading to  a bound on the success probabilities of the adversary in the games. Combining these bounds via the triangle inequality shows that the overall advantage against pHE-CD is at most twice the iKEM advantage plus the DEM-CD advantage, which is negligible under the stated assumptions; hence pHE-CD is EV-$q_e$-CD secure. We now give the full proof.

\noindent \textbf{Notations for challenger and adversary.} Since the theorem is stated for both the computational and the information-theoretic case, we let $\mathcal{C}$ denote the computational class (QPT, or computationally unbounded) for which the EV-$q_e$-CD security of {\PCD} is claimed in a given instance of the theorem; that is, $\mathcal{C}$ is QPT when {\SD} is computationally one-time EV-CD secure (Construction~2, Section~\ref{sec: construction}), and $\mathcal{C}$ is the class of computationally unbounded algorithms when {\SD} is information-theoretically one-time EV-CD secure (Construction~1, Section~\ref{sec: construction}).
\begin{itemize}
    \item \textbf{pHE-CD}: The challenger is denoted $ch$; the adversary is $\mathcal{A} = (\mathcal{A}_1, \mathcal{A}_2)$, where $\mathcal{A}_1 \in \mathcal{C}$ and $\mathcal{A}_2$ is an adversary with unbounded quantum computation (as required by Definition~\ref{Def: pHE-CD sec}).
    \item \textbf{iKEM}: The challenger is denoted $ch^\text{iK}$; the adversary is $\mathcal{A}^\text{iK} = (\mathcal{A}^\text{iK}_1, \mathcal{A}^\text{iK}_2)$. Since {\sf iK} security is unconditional (Definition~\ref{Def:iKEM}), both $\mathcal{A}^\text{iK}_1$ and $\mathcal{A}^\text{iK}_2$ are computationally unbounded.
    \item \textbf{DEM-CD}: The challenger is denoted $ch^\text{D-CD}$; the adversary is $\mathcal{A}^\text{D-CD}=(\mathcal{A}_1^\text{D-CD}, \mathcal{A}_2^\text{D-CD})$, where $\mathcal{A}_1^\text{D-CD} \in \mathcal{C}$ and $\mathcal{A}_2^\text{D-CD}$ is computationally unbounded.
\end{itemize}
As shown below, both reductions we construct, $\mathcal{A}^\text{iK}$ and $\mathcal{A}^\text{D-CD}$, run $\mathcal{A}_1$ as a black-box subroutine without performing any further computation besides the bookkeeping needed to relay oracle queries and ciphertexts. Consequently $\mathcal{A}_1^{\text{iK}}$ and $\mathcal{A}_1^{\text{D-CD}}$ belong to the same class $\mathcal{C}$ as $\mathcal{A}_1$, so the reduction is meaningful, and tight, in both the computational and the information-theoretic case; this is what allows the single argument below to establish both parts of the theorem.

\begin{proof}
To prove the theorem, we show that for any adversaries $\mathcal{A}$, $\mathcal{A}^\text{iK}$, and $\mathcal{A}^\text{D-CD}$ as defined above, attacking {\PCD}, {\sf iKEM}, and {\sf DEM-CD} respectively,
\begin{align*}
    \text{Adv}_{\text{\PCD}, \mathcal{A}}^{\text{ev-}q_e \text{-cd}}(\lambda) \le 2\text{Adv}_{\text{iK}, \mathcal{A}^\text{iK}}^{\text{ind-}q_e \text{-cea}}(\lambda)+\text{Adv}_{\text{\SD}, \mathcal{A}^\text{\SD}}^{\text{ev-cd}}(\lambda).
\end{align*}
To show this we consider the following two games, $G_0(\lambda, b)$ and $G_1(\lambda, b)$.\\

\noindent $G_0(\lambda, b)$: This game is the same as $\text{Exp}_{\text{\PCD}, \mathcal{A}}^{\text{ev-}q_e \text{-cd}}(\lambda, b)$.\\

\noindent $G_1(\lambda, b)$: This is identical to $G_0(\lambda, b)$, except that in Step~4, the key generated by the iKEM is replaced by an independently sampled random key:
\begin{enumerate}
    \item[$4^\prime$.] The challenger computes $(K, C_1) \leftarrow {\sf iK.Encap}(X)$. It then chooses $\hat{K} \leftarrow \mathcal{K}$ and evaluates $(C_2,vk) \leftarrow \text{\sf \SD.Encap}(m_b,\hat{K})$. Sets $CT = (C_1,C_2)$ and sends $CT$ to $\mathcal{A}$.
\end{enumerate}

\begin{proposition}
\label{prop: iKEM}
    If {\sf iKEM} {\sf iK} is $\text{IND-}q_e \text{-CEA}$ secure then,
    \[\left|\Pr\left[G_0(\lambda, b) = 1 \right] - \Pr\left[G_1(\lambda, b) = 1 \right]\right| \leq \text{Adv}_{\text{iK}, \mathcal{A}^\text{iK}}^{\text{ind-}q_e \text{-cea}}(\lambda)\]
\end{proposition}
\begin{proof}
  If the adversary $\mathcal{A}$ can distinguish the games $G_0(\lambda,b)$ and $G_1(\lambda,b)$ with non-negligible advantage, we construct an adversary $\mathcal{A}^{\text{iK}}$ against {\sf iK} that breaks its $\text{IND-}q_e\text{-CEA}$ security with non-negligible advantage, using $\mathcal{A}$ as a subroutine. Here $\mathcal{A}^{\text{iK}}$ acts as the challenger for $\mathcal{A}$.

   \noindent \textbf{Construction of} $\mathcal{A}^{\text{iK}}$:
  \begin{itemize}
       \item $\mathcal{A}_1^{\text{iK}}$ receives $Z$ from the challenger $ch^\text{iK}$ of $\text{Exp}_{\text{iK}, \mathcal{A}^{\text{iK}}}^{\text{ind-}q_e \text{-cea}}(\lambda, b'')$, $b'' \in \{0,1\}$, and initializes the adversary $\mathcal{A}_1$ with the random variable $Z$.
       \item When $\mathcal{A}_1$ issues an encryption query on a message $m$, $\mathcal{A}_1^{\text{iK}}$ queries the encapsulation oracle to obtain a key--ciphertext pair $(K, C_1)$. It then uses the key $K$ to compute $(vk, C_2) \leftarrow \textsf{\SD.Encap}(K, m)$ and responds to $\mathcal{A}_1$'s query with $(vk, CT)$, where $CT = (C_1, C_2)$.
       \item At some point $\mathcal{A}_1^{\text{iK}}$ (the challenger of the simulated game) receives two messages $m_0$ and $m_1$ from $\mathcal{A}_1$. It asks for the challenge ciphertext--key pair from its own challenger and passes its state to $\mathcal{A}_2^{\text{iK}}$.
      \item According to a randomly chosen bit $b$, $\mathcal{A}_2^{\text{iK}}$ computes $(C_2^*,vk) \leftarrow \text{\sf \SD.{Encap}}(K^* , m_b)$, and sends $CT^*=(C_1^*, C_2^*)$ to $\mathcal{A}_1$ as the challenge ciphertext.
       \item At some point, $\mathcal{A}_1$ outputs a deletion certificate $cert$ together with an internal (quantum) state $\sigma$. The certificate $cert$ is forwarded to $\mathcal{A}_2^{\text{iK}}$, and the internal state of the adversary is passed to $\mathcal{A}_2$.
       \item $\mathcal{A}_2^{\text{iK}}$ runs the verification algorithm $\textsf{\SD.Vrfy}(vk, \mathit{cert})$. If the output is $\bot$, it provides $\bot$ to $\mathcal{A}_2$; otherwise, if the output is $\top$, it provides the key $K^*$ to $\mathcal{A}_2$.
       \item $\mathcal{A}_2$ outputs its guess bit $b' \in \{0, 1\}$ (for the chosen bit of $\mathcal{A}^{\text{iK}}$).
       \item Finally, $\mathcal{A}_2^{\text{iK}}$ outputs the same bit $b'$ to its own challenger.
   \end{itemize}
Note that if $b'' = 0$, ${CT}^* = (C_1^*, C_2^*)$ where $(C_1^*,K) = {\sf iK.Encap}(X)$ and $(C_2^*,vk)\leftarrow \text{\sf \SD.{Encap}}(m_b, K)$; hence $\mathcal{A}^{\text{iK}}$ perfectly simulates $G_0(\lambda,b)$. Similarly, if $b'' = 1$, $(C_2^*,vk) \leftarrow \text{\sf \SD.{Encap}}(m_b, \hat{K})$ for $\hat{K} \stackrel{\$}\gets \mathcal{K}$, so $\mathcal{A}^{\text{iK}}$ perfectly simulates $G_1(\lambda,b)$. Throughout, $\mathcal{A}_1^{\text{iK}}$ and $\mathcal{A}_2^{\text{iK}}$ perform no computation beyond running $\mathcal{A}_1, \mathcal{A}_2$ as a subroutine and relaying messages, so $\mathcal{A}^{\text{iK}} \in \mathcal{C}$ whenever $\mathcal{A} \in \mathcal{C}$. Hence,
    \[\left|\Pr\left[G_0(\lambda, b) = 1 \right] - \Pr\left[G_1(\lambda, b) = 1 \right]\right| \le \text{Adv}_{\text{iK}, \mathcal{A}^\text{iK}}^{\text{ind-}q_e \text{-cea}}(\lambda)\]
    Therefore, if $\mathcal{A}$ can distinguish between $G_0(\lambda, b)$ and $G_1(\lambda, b)$, then $\mathcal{A}^{\text{iK}}$ can break the IND-$q_e$-CEA security of {\sf iK}.
\end{proof}
\begin{proposition}
\label{prop: DEMCD}
    If DEM-CD scheme {\SD} is one-time EV-CD secure then,
    \[\left|\Pr\left[G_1(\lambda, 0) = 1 \right] - \Pr\left[G_1(\lambda, 1) = 1 \right]\right| \leq \text{Adv}_{\text{\SD}, \mathcal{A}^\text{\SD}}^{\text{ev-cd}}(\lambda)\]
\end{proposition}
\begin{proof}
   We construct an adversary $\mathcal{A}^\text{\SD}$ that breaks the one-time EV-CD security of {\sf \SD}, assuming $\mathcal{A}$ can distinguish between $G_1(\lambda, 0)$ and $G_1(\lambda, 1)$. Here $\mathcal{A}^{\text{\SD}}$ simulates the environment of the experiment $G_1(\lambda, b)$ and acts as the challenger for $\mathcal{A}$.

\noindent \textbf{Construction of} $\mathcal{A}^\text{\SD}$: $\mathcal{A}^\text{\SD}$ takes part in the experiment $\text{Exp}_{\text{\SD}, \mathcal{A}^\text{\sf \SD}}^{\text{ev-cd}}(\lambda, b)$:
\begin{itemize}
    \item $\mathcal{A}^\text{\SD}_1$ generates $(X,Y,Z) \leftarrow \text{\sf iK.Gen}(1^\lambda,P)$ and initiates $\mathcal{A}_1$ with $Z$.
    \item $\mathcal{A}_1$ outputs two messages $m_0,m_1$ to $\mathcal{A}^\text{\SD}_1$.
    \item $\mathcal{A}^\text{\SD}_1$ sends $m_0,m_1$ to its challenger $ch^\text{\SD}$.
    \item $ch^\text{\SD}$ chooses a random key $K \in \mathcal{K}$, computes $(C_2^*,vk) \leftarrow \text{\sf \SD.Encap}(K,m_b)$, and sends $C^*_2$ to $\mathcal{A}^\text{\SD}_1$.
    \item $\mathcal{A}^\text{\SD}_1$ generates $(C^*_1, K') \leftarrow {\sf iK.Encap}(X)$ and sends $(C^*_1, C^*_2)$ to $\mathcal{A}_1$.
    \item At some point, $\mathcal{A}_1$ sends a certificate $cert$ to $\mathcal{A}^\text{\SD}_1$, which forwards it to its challenger $ch^\text{\SD}$. $\mathcal{A}^\text{\SD}_1$ then passes its internal state to $\mathcal{A}^\text{\SD}_2$.
    \item If $cert$ is valid, $\mathcal{A}^\text{\SD}_2$ receives $K$ from $ch^\text{\SD}$, and forwards $K$ to $\mathcal{A}_2$.
    \item $\mathcal{A}_2$ outputs a bit $b' \in \{0,1\}$.
    \item $\mathcal{A}^\text{\SD}_2$ outputs whatever $\mathcal{A}_2$ outputs.
\end{itemize}
Note that the key $K'$ produced by $\text{\sf iK.Encap}$ is discarded and plays no role beyond producing a well-formed $C_1^*$, since $\mathcal{A}_1$'s view depends on the iKEM ciphertext but not on the key it encapsulates; this is exactly what makes $G_1$ (where the DEM-CD key is independent of the iKEM-encapsulated key) the right intermediate game to reduce to {\SD}'s security. When $b=0$, $C_2^*= \text{\SD.{Encap}}(K,m_0)$, so $\mathcal{A}^\text{\SD}$ perfectly simulates $G_1(\lambda,0)$; similarly, when $b=1$, $\mathcal{A}^\text{\SD}$ perfectly simulates $G_1(\lambda,1)$. As before, $\mathcal{A}_1^{\text{\SD}}$ relays $\mathcal{A}_1$ as a black box, so $\mathcal{A}_1^{\text{\SD}} \in \mathcal{C}$ whenever $\mathcal{A}_1 \in \mathcal{C}$, and $\mathcal{A}_2^{\text{\SD}}$ is unbounded as required by the definition of one-time EV-CD security. Therefore,
  \[\left|\Pr\left[G_1(\lambda, 0) = 1 \right] - \Pr\left[G_1(\lambda, 1) = 1 \right]\right| \le \text{Adv}_{\text{\SD}, \mathcal{A}^\text{\SD}}^{\text{ev-cd}}(\lambda)\]
In other words, if $\mathcal{A}$ can distinguish between $G_1(\lambda,0)$ and $G_1(\lambda,1)$, then $\mathcal{A}^\text{\sf \SD}$ can break one-time EV-CD security of {\SD}.
\end{proof}

From Propositions~\ref{prop: iKEM} and \ref{prop: DEMCD} we have,
\begin{align*}
\begin{multlined}[t]
\left|\Pr[G_0(\lambda,0)=1]-\Pr[G_0(\lambda,1)=1]\right| \\
\le\left|\Pr[G_0(\lambda,0)=1]-\Pr[G_1(\lambda,0)=1]\right| \\
+\left|\Pr[G_1(\lambda,0)=1]-\Pr[G_1(\lambda,1)=1]\right| \\
+\left|\Pr[G_1(\lambda,1)=1]-\Pr[G_0(\lambda,1)=1]\right|
\end{multlined} \\
\begin{multlined}[t]
=2\text{Adv}_{\text{iK},\mathcal{A}^\text{iK}}^{\text{ind-}q_e\text{-cea}}(\lambda)
+\text{Adv}_{\text{\SD},\mathcal{A}^\text{\SD}}^{\text{ev-cd}}(\lambda)
\end{multlined}
\le\text{negl}(\lambda).
\end{align*}
Since {\sf iK} is IND-$q_e$-CEA secure, $\text{Adv}_{\text{iK}, \mathcal{A}^\text{iK}}^{\text{ind-}q_e \text{-cea}}(\lambda) \le \text{negl}(\lambda)$ unconditionally, i.e., for $\mathcal{A}^\text{iK}$ ranging over all, possibly unbounded, adversaries. When {\SD} is one-time EV-CD secure against $\mathcal{C}=$ QPT adversaries (Construction~2), $\text{Adv}_{\text{\SD}, \mathcal{A}^\text{\SD}}^{\text{ev-cd}}(\lambda) \le \text{negl}(\lambda)$ holds for $\mathcal{A}^\text{\SD} \in \mathcal{C}$, giving computational EV-$q_e$-CD security of {\PCD}; when {\SD} is one-time EV-CD secure against $\mathcal{C}=$ unbounded adversaries (Construction~1), the same bound holds for unbounded $\mathcal{A}^\text{\SD}$, giving information-theoretic EV-$q_e$-CD security of {\PCD}. In either case, $\left|\Pr\left[G_0(\lambda, 0) = 1 \right] - \Pr\left[G_0(\lambda, 1)= 1 \right]\right| \le \text{negl}(\lambda)$, which implies that {\PCD} is EV-$q_e$-CD secure.
\end{proof}

\section{Constructions}
\label{sec: construction}
\vspace{-1mm}
In this section, we present two constructions: (1) a pHE-CD scheme achieving information-theoretic security for both encryption and certified deletion, and (2) a pHE-CD scheme providing computational security for encryption 
and everlasting certified deletion security.

\noindent \emph{Construction 1:}
It composes
(i) an IND-$q_e$-CEA–secure information-theoretic iKEM {\sf iK} = ({\sf iK.KeyGen, iK.Encap, iK.Decap}), formalized in
~\cite{sharifian2021information},
and (ii) an information-theoretically
IND-OT secure
DEM-CD \cite[Definition 2.11]{broadbent2020quantum}.
The iKEM can be realized using a one-way (one-message)
key agreement (SKA) protocol in correlated randomness setting \cite{sharifian2021information}, and the DEM-CD
can be instantiated directly by the one-time SKE-CD scheme of Broadbent and Islam (BI)~\cite{broadbent2020quantum}, which we denote by {\sf \SD} = ({\sf \SD.KeyGen}, {\sf \SD.Encap}, {\sf \SD.Decap}, {\sf \SD.Del}, {\sf \SD.Vrfy}) and treat as a black box satisfying the syntax of Definition~\ref{Def:SKECD}. BI's scheme provides information-theoretic IND-OT encryption security \cite[Definition~2.11]{broadbent2020quantum} and, as its main contribution, unconditional one-time certified deletion security~\cite{broadbent2020quantum}, which is exactly the notion captured by one-time EV-CD security (Definition~\ref{Def: SKE-CD EV-CD}) when both $\mathcal{A}_1$ and $\mathcal{A}_2$ are computationally unbounded. We refer the reader to \cite{broadbent2020quantum} for the internal construction of {\SD} (based on Wiesner coding combined with a one-time pad on the message and an authentication string on the encoding basis); unlike the one-bit scheme constructed for Construction~2 below, {\SD} here directly supports the unrestricted message space $\mathcal{M} = \{0,1\}^*$, so Construction~1 requires no bit-by-bit repetition.

Construction~1, {\sf \PCD} = ({\sf \PCD.KeyGen}, {\sf \PCD.Enc}, {\sf \PCD.Dec}, {\sf \PCD.Del}, {\sf \PCD.Vrfy}), is the direct instantiation of the generic pHE-CD construction of Figure~\ref{fig:pHE-CD} with the above {\sf iK} and {\SD}:
\begin{enumerate}
    \item {\sf \PCD.KeyGen}$(1^\lambda,P)\rightarrow (X,Y,Z)$:
    \begin{enumerate}
        \item Generate $(X,Y,Z) \leftarrow \text{\sf iK.KeyGen}(1^\lambda)$.
    \end{enumerate}
    \item {\sf \PCD.Enc}$(X, m) \rightarrow (vk, CT)$, for $m \in \{0,1\}^*$:
    \begin{enumerate}
        \item Generate $(K, C_{1}) \leftarrow \text{\sf iK.Encap}(X)$.
        \item Compute $(vk, C_{2}) \leftarrow \text{\sf \SD.Encap}(K, m)$.
        \item Output $CT= (C_1,C_2)$ and $vk$.
    \end{enumerate}
    \item {\sf \PCD.Dec}$(Y, CT) \rightarrow m \text{ or } \bot$:
    \begin{enumerate}
        \item Parse $CT= (C_{1},C_{2})$.
        \item Compute $K \leftarrow {\sf iK.Decap}(Y,C_{1})$.
        \item If $K = \bot$, return $\bot$. Else, compute $m \leftarrow \text{\sf \SD.Decap}(K,C_{2})$ and output $m$.
    \end{enumerate}
    \item {\sf \PCD.Del}$({CT}) \rightarrow cert$:
    \begin{enumerate}
        \item Parse $CT= (C_{1},C_{2})$.
        \item Compute $cert \leftarrow {\sf \SD.Del}(C_2)$.
    \end{enumerate}
    \item {\sf \PCD.Vrfy}$(vk, cert) \rightarrow \top \text{ or } \bot$:
    \begin{enumerate}
        \item Output $\text{\sf \SD.Vrfy}(vk, cert)$.
    \end{enumerate}
\end{enumerate}

\noindent \textbf{Security of Construction~1:}
\begin{theorem}
\label{Thm: Construction1: Enc}
    If {\sf iK} is information-theoretically $\text{IND-}q_e\text{-CEA}$ secure and {\SD} is information-theoretically IND-OT secure, then {\PCD} (Construction~1) achieves information-theoretic $\text{IND-}q_e\text{-CPA}$ encryption security.
\end{theorem}
The proof follows directly from Theorem~\ref{Thm: pHECD encryption}, taking both {\sf iK} and {\SD} to be information-theoretically secure.
\begin{theorem}
\label{Thm: Construction1: CD}
    If {\sf iK} is $\text{IND-}q_e\text{-CEA}$ secure and {\SD} is one-time EV-CD secure with information-theoretic security (as provided by BI's scheme~\cite{broadbent2020quantum}), then {\PCD} (Construction~1) achieves $\text{EV-}q_e\text{-CD}$ security against a computationally unbounded two-stage adversary $\mathcal{A} = (\mathcal{A}_1,\mathcal{A}_2)$; that is, Construction~1 achieves information-theoretic certified deletion security.
\end{theorem}
\begin{proof}
This is precisely the information-theoretic case ($\mathcal{C}=$ computationally unbounded algorithms) of Theorem~\ref{Thm: pHECD compose}, instantiated with the {\sf iK} and {\SD} of Construction~1. Revisiting the proof of Theorem~\ref{Thm: pHECD compose} with $\mathcal{C}$ set to the unbounded class: the reduction $\mathcal{A}^\text{iK}$ of Proposition~\ref{prop: iKEM} relays $\mathcal{A}_1$ as a black box and so is unbounded whenever $\mathcal{A}_1$ is, and $\text{Adv}_{\text{iK}, \mathcal{A}^\text{iK}}^{\text{ind-}q_e \text{-cea}}(\lambda) \le \text{negl}(\lambda)$ holds unconditionally, against any (possibly unbounded) $\mathcal{A}^\text{iK}$, since {\sf iK} is information-theoretically secure. Likewise, the reduction $\mathcal{A}^\text{\SD}$ of Proposition~\ref{prop: DEMCD} relays $\mathcal{A}_1$ as a black box and so is unbounded whenever $\mathcal{A}_1$ is, and since {\SD} (i.e., BI's scheme) is one-time EV-CD secure unconditionally, $\text{Adv}_{\text{\SD}, \mathcal{A}^\text{\SD}}^{\text{ev-cd}}(\lambda) \le \text{negl}(\lambda)$ holds against an unbounded $\mathcal{A}^\text{\SD}$ as well. Hence
{\small
\[
\begin{aligned}
\left|\Pr\left[G_0(\lambda, 0) = 1 \right]
- \Pr\left[G_0(\lambda, 1)= 1 \right]\right|
&\le 2\text{Adv}_{\text{iK}, \mathcal{A}^\text{iK}}^{\text{ind-}q_e\text{-cea}}(\lambda)
+ \text{Adv}_{\text{SD}, \mathcal{A}^\text{SD}}^{\text{ev-cd}}(\lambda) \\
&\le \text{negl}(\lambda).
\end{aligned}
\]
}
holds against any, possibly unbounded, two-stage adversary $\mathcal{A} = (\mathcal{A}_1,\mathcal{A}_2)$, which is precisely $\text{EV-}q_e\text{-CD}$ security of {\PCD} against an unbounded $\mathcal{A}_1$, i.e., information-theoretic certified deletion security.
\end{proof}
An important difference 
between 
the two constructions is that Construction 2 in its basic form is for a one-bit message (and can be extended to multi-bit messages by repeated application to each bit of the message),  
while Construction 1 is for multi-bit message.
%


\noindent \emph{Construction 2.} The construction uses i) an IND-$q_e$-CEA secure iKEM, {\sf iK} = ({\sf iK.KeyGen, iK.Encap, iK.Decap}) (see ~\cite{sharifian2021information} for definition and instantiation); and ii) a  computational  IND-OT secure
DEM, {\sf D = (D.KeyGen, D.{Encap}, D.Decap)}.

 We first construct a new computationally secure (quantum-resistant) DEM-CD (see below) for a one-bit message, and then use it to  construct a  pHE-CD for one-bit messages.
A multi-bit  pHE-CD can be constructed by 
%
encrypting each message bit independently using the single-bit pHE-CD  
and concatenating the resulting ciphertexts.

\noindent

{\textbf{Construction of a one-bit
 computationally secure {\SD} (everlasting deletion security):}}

\begin{enumerate}
    \item {\sf \SD.KeyGen}$(1^\lambda)\rightarrow K$: 
    \begin{enumerate}
        \item $K \leftarrow {\sf D.KeyGen}(1^\lambda)$
    \end{enumerate}
    \item {\sf \SD.Encap}$(K,m) \rightarrow (C,vk)$: 
    \begin{enumerate}
    \item Choose $x, \theta \leftarrow \{0,1\}^\lambda$.
       \item Compute $C = (\ket{x}_{\theta}, {\sf D.Encap}(K, (\theta, m \oplus \bigoplus_{i:\theta_{i}=0} x_{i})$ and set $vk=(x,\theta)$.
       \end{enumerate}
    \item {\sf \SD.Decap}$({C},K) \rightarrow m$:
    \begin{enumerate}
        \item Parse $C = (\ket{x}_{\theta}, C')$.
        \item Computes $(\theta,m') \leftarrow \text{\sf D.Decap}(K,C')$, measure $\ket{x}_{\theta}$ in the $\theta$    basis to obtain $x$ and outputs $m = m' \oplus \bigoplus_{i:\theta_i=0}x_i$ .
    \end{enumerate}
    \item {\sf \SD. {Del}}$(C) \rightarrow cert$: 
    \begin{enumerate}
        \item Parse $C = (\ket{x}_{\theta}, C')$. 
        \item Measure $\ket{x}_{\theta}$ in the Hadamard basis to obtain the string $cert \in \{0,1\}^\lambda$.
        
    \end{enumerate}
    
    \item {\sf \SD.Vrfy}$(vk, cert) \rightarrow \top \text{ or } \bot$: 
    \begin{enumerate}

        \item Parse $vk=(x,\theta)$ and $cert$ as $x'$.
        \item Output $\top$ if for all $i \in [\lambda]$, $cert_i = x_i'$.

    \end{enumerate}
     \end{enumerate}

\begin{theorem}
\label{Thm: DEM-CD: CD}
If {\sf D} is a DEM satisfying computationally IND-OT security, 
then 
{\SD} is a DEM-CD that  achieves computationally IND-OT encryption security and one-time EV-CD deletion security.
\end{theorem}
\vspace{-2mm}

The proof strategy for deletion security is inspired by the certified deletion security proof of the one-bit
PKE-CD construction in \cite{bartusek2023cryptography}: we use a computationally secure DEM {\sf D}
to encrypt
(data-encapsulate) the randomness used in quantum encoding,
under the shared key. We give the full proof below, establishing the two claims of the theorem separately.

\begin{proof}
\noindent \emph{Encryption security.} We construct a QPT adversary $\mathcal{A}'$ against {\sf D} from any QPT adversary $\mathcal{A}$ against the IND-OT encryption security of {\SD}, such that $\text{Adv}_{\text{\SD},\mathcal{A}}^{\text{ind-ot}}(\lambda) = \text{Adv}_{\text{D},\mathcal{A}'}^{\text{ind-ot}}(\lambda)$. In ${\sf \SD.Encap}(K,m)$, the strings $x,\theta \leftarrow \{0,1\}^\lambda$ are sampled independently of $m$, and the only quantity passed to ${\sf D.Encap}$ is the pair $(\theta,\, m\oplus \bigoplus_{i:\theta_i=0}x_i)$; the qubit register $\ket{x}_\theta$ carries no further information about $m$, since $x$ itself is independent of $m$. Given $\mathcal{A}$'s challenge messages $m_0,m_1$, the reduction $\mathcal{A}'$ samples $x,\theta \leftarrow \{0,1\}^\lambda$ itself, forms $M_0 = (\theta, m_0\oplus \bigoplus_{i:\theta_i=0}x_i)$ and $M_1 = (\theta, m_1\oplus \bigoplus_{i:\theta_i=0}x_i)$, and submits $(M_0,M_1)$ to its own challenger. On receiving $C' \leftarrow {\sf D.Encap}(K,M_b)$, $\mathcal{A}'$ returns $C = (\ket{x}_\theta, C')$ to $\mathcal{A}$. Since $x,\theta$ are sampled independently of $b$, this is a perfect simulation of the {\SD} IND-OT experiment with bit $b$, and $\mathcal{A}'$ performs no further computation beyond sampling $(x,\theta)$ and running $\mathcal{A}$, so $\mathcal{A}'$ is QPT whenever $\mathcal{A}$ is. Hence $\text{Adv}_{\text{\SD},\mathcal{A}}^{\text{ind-ot}}(\lambda) = \text{Adv}_{\text{D},\mathcal{A}'}^{\text{ind-ot}}(\lambda) \le \text{negl}(\lambda)$, since {\sf D} is computationally IND-OT secure.

\noindent \emph{Deletion security.}
To prove the theorem, we follow the
proof strategy of
\cite[Theorem~3.1]{bartusek2023cryptography}.
To establish the one-time EV-CD security of {\SD} (Definition~\ref{Def: SKE-CD EV-CD}) conditioned on the
validity of the deletion certificate,
we define a quantum operation $\mathcal{Z}_\lambda(\cdot,\cdot,\cdot)$ with three inputs as follows:
\begin{enumerate}
    \item Sample $K \leftarrow {\sf D.KeyGen}(1^\lambda)$, $x, \theta \leftarrow \{0,1\}^\lambda$. 
    \item Compute $m'= m\oplus \bigoplus_{i:\theta_i = 0} x_i$ and the BB84 state $\ket{x}_\theta$. 
    \item Output $(\ket{x}_\theta,{\sf D.Encap}(K,(\theta,m')))$.
\end{enumerate}

Since
{\sf D} is an IND-OT secure DEM, it admits no encryption queries beyond the challenge. By Theorem~3.1 of Bartusek et al.~\cite{bartusek2023cryptography} together with the IND-OT security of {\sf D} (Definition~\ref{def:dem}), we
obtain,
  \[ \text{Adv}_{\text{\SD}, \mathcal{A}}^{\text{ev-cd}}(\lambda) \le \text{negl}(\lambda) \]
See Appendix~\ref{Appendix 1} for a restatement of the theorem {\cite[Theorem 3.1]{bartusek2023cryptography}} and a discussion of how the above inequality is obtained. Therefore {\SD} achieves one-time EV-CD security.
\end{proof}

{\bf A computationally secure \PCD~(everlasting deletion security)}
\begin{enumerate}
    \item {\sf \PCD.KeyGen}$(1^\lambda,P)\rightarrow (X,Y,Z)$: 
    \begin{enumerate}
        \item Generate $(X,Y,Z) \leftarrow \text{\sf iK.KeyGen}(1^\lambda)$.

    \end{enumerate}
     
    \item {\sf \PCD. Enc}$(X, m) \rightarrow (vk, CT)$: 
    \begin{enumerate}
    \item Generate $({K}, C_{1}) \leftarrow \text{\sf iK.Encap}(X)$.
        \item Choose $x, \theta \leftarrow \{0,1\}^\lambda$.
       \item Compute $C_{2} = (\ket{x}_{\theta}, {\sf D.Encap}(K, (\theta, m \oplus \bigoplus_{i:\theta_{i}=0} x_{i})$ and set $vk=(x,\theta)$.
      \item Output $CT= (C_1,C_2)$ and $vk$.
    
    \end{enumerate}
    
    \item {\sf \PCD.{Dec}}$(Y, CT) \rightarrow m \text{ or } \bot$: 
    \begin{enumerate}
        \item Parse $CT= (C_{1},C_{2})$ where $C_{2} = (\ket{x}_{\theta}, C'_{2})$.
        \item Computes $K \leftarrow {\sf iK.Decap}(C_{1},Y)$.
        \item If $K = \bot$; return $\bot$. Else computes $(\theta,m') \leftarrow \text{\sf D.Decap}(K,C_{2})$, measure $\ket{x}_{\theta}$ in the $\theta$    basis to obtain $x$ and outputs $m = m' \oplus \bigoplus_{i:\theta_i=0}x_i$. 
    \end{enumerate}
    \item {\sf \PCD. {Del}}$({CT}) \rightarrow cert$: 
    \begin{enumerate}
        \item Parse $CT= (C_{1},C_{2})$ where $C_{2} = (\ket{x}_{\theta}, C'_{2})$. 
        \item Measure $\ket{x}_{\theta}$ in the Hadamard basis to obtain the string $cert \in \{0,1\}^\lambda$.
        
    \end{enumerate}
    
    \item {\sf \PCD.Vrfy}$(vk, cert) \rightarrow \top \text{ or } \bot$: 
    \begin{enumerate}

        \item Parse $vk=(x,\theta)$ and $cert$ as $x'$.
        \item Output $\top$ if for all $i \in [\lambda]$, $cert_i = x_i'$.

    \end{enumerate}
\end{enumerate}

\noindent 

\noindent \textbf{Security of Construction~2:}
\vspace{-2mm}
\begin{theorem}
\label{Thm: Construction2: Enc}
    If {\sf iK} is information-theoretically $\text{IND-}q_e \text{-CEA}$ secure and {\sf D} is computationally IND-OT secure, then {\PCD} achieves computationally $\text{IND-}q_e \text{-CPA}$ encryption security.
\end{theorem}
\vspace{-2mm}
The proof follows directly from Theorem~\ref{Thm: pHECD encryption} and Theorem~\ref{Thm: DEM-CD: CD} by noting that the {\sf iK} and {\sf D} satisfy the requirements of the theorem.
\vspace{-1mm}
\begin{theorem}
\label{Thm: Construction2: CD}
    If {\sf iK} is information-theoretically $\text{IND-}q_e \text{-CEA}$ secure and {\sf D} is computationally IND-OT secure, then {\PCD} achieves $\text{EV-}q_e \text{-CD}$ security.
\end{theorem}
\begin{proof}
\vspace{-2mm}
The claim follows immediately as a corollary of
Theorem~\ref{Thm: pHECD compose} and Theorem~\ref{Thm: DEM-CD: CD}.
In particular, we may assume that $C_2 = \bigl(\ket{x}_{\theta},
{\sf D.Encap}\bigl(K,(\theta, m \oplus \bigoplus_{i:\theta_i = 0} x_i)\bigr)\bigr)
= \text{\sf \SD.Encap}(K,m).$ By Theorem~\ref{Thm: DEM-CD: CD}, it follows that {\SD} achieves one-time
EV-CD security. Applying Theorem~\ref{Thm: pHECD compose}, we conclude that
if {\sf iK} is $\text{IND-}q_e\text{-CEA}$ secure and {\SD} is one-time
EV-CD secure, then {\PCD} is $\text{EV-}q_e\text{-CD}$ secure.
This completes the proof.
\vspace{-2mm}
\end{proof}
\vspace{-2mm}
\section{Concluding Remarks}
\label{conclusion}
\vspace{-1mm}
We introduced and formalized 
pHE-CD and 
gave two constructions with proved security. Both constructions allow Bob to re-use his key, and the second construction uses constant-length key.

The quantum codings of these constructions are Wiesner conjugate coding \cite{wiesner1983conjugate} that can be efficiently implemented using today's quantum technologies.
The security of our construction is independent on computational hardness assumptions and instead relies on fundamental quantum principles, standard physical assumptions, and the optimality of quantum algorithms for unstructured search (e.g., Grover’s bound \cite{boyer1998tight,grover1996fast}), providing robustness against future computational advances.
The pHE-CD framework enables constructions based on QKD-derived iKEMs; extending qKEM security \cite{11195563} to adversaries with encapsulation-oracle access is left for future work.

\bibliographystyle{IEEEtran}
\bibliography{ref}

\begin{thebibliography}{10}
\providecommand{\url}[1]{#1}
\csname url@samestyle\endcsname
\providecommand{\newblock}{\relax}
\providecommand{\bibinfo}[2]{#2}
\providecommand{\BIBentrySTDinterwordspacing}{\spaceskip=0pt\relax}
\providecommand{\BIBentryALTinterwordstretchfactor}{4}
\providecommand{\BIBentryALTinterwordspacing}{\spaceskip=\fontdimen2\font plus
\BIBentryALTinterwordstretchfactor\fontdimen3\font minus \fontdimen4\font\relax}
\providecommand{\BIBforeignlanguage}[2]{{%
\expandafter\ifx\csname l@#1\endcsname\relax
\typeout{** WARNING: IEEEtran.bst: No hyphenation pattern has been}%
\typeout{** loaded for the language `#1'. Using the pattern for}%
\typeout{** the default language instead.}%
\else
\language=\csname l@#1\endcsname
\fi
#2}}
\providecommand{\BIBdecl}{\relax}
\BIBdecl

\bibitem{dieks1982communication}
D.~Dieks, ``Communication by epr devices,'' \emph{Physics Letters A}, vol.~92, no.~6, pp. 271--272, 1982.

\bibitem{wootters1982single}
W.~K. Wootters and W.~H. Zurek, ``A single quantum cannot be cloned,'' \emph{Nature}, vol. 299, no. 5886, pp. 802--803, 1982.

\bibitem{park1970concept}
J.~L. Park, ``The concept of transition in quantum mechanics,'' \emph{Foundations of physics}, vol.~1, no.~1, pp. 23--33, 1970.

\bibitem{broadbent2020quantum}
A.~Broadbent and R.~Islam, ``Quantum encryption with certified deletion,'' in \emph{Theory of Cryptography: 18th International Conference, TCC 2020, Durham, NC, USA, November 16--19, 2020, Proceedings, Part III 18}.\hskip 1em plus 0.5em minus 0.4em\relax Springer, 2020, pp. 92--122.

\bibitem{unruh2015revocable}
D.~Unruh, ``Revocable quantum timed-release encryption,'' \emph{Journal of the ACM (JACM)}, vol.~62, no.~6, pp. 1--76, 2015.

\bibitem{wiesner1983conjugate}
S.~Wiesner, ``Conjugate coding,'' \emph{ACM Sigact News}, vol.~15, no.~1, pp. 78--88, 1983.

\bibitem{bennett2014quantum}
C.~H. Bennett and G.~Brassard, ``Quantum cryptography: Public key distribution and coin tossing,'' \emph{Theoretical computer science}, vol. 560, pp. 7--11, 2014.

\bibitem{tomamichel2012tight}
M.~Tomamichel, C.~C.~W. Lim, N.~Gisin, and R.~Renner, ``Tight finite-key analysis for quantum cryptography,'' \emph{Nature communications}, vol.~3, no.~1, p. 634, 2012.

\bibitem{bacco2013experimental}
D.~Bacco, M.~Canale, N.~Laurenti, G.~Vallone, and P.~Villoresi, ``Experimental quantum key distribution with finite-key security analysis for noisy channels,'' \emph{Nature communications}, vol.~4, no.~1, p. 2363, 2013.

\bibitem{diamanti2016practical}
E.~Diamanti, H.-K. Lo, B.~Qi, and Z.~Yuan, ``Practical challenges in quantum key distribution,'' \emph{npj Quantum Information}, vol.~2, no.~1, pp. 1--12, 2016.

\bibitem{hayashi2006practical}
M.~Hayashi, ``Practical evaluation of security for quantum key distribution,'' \emph{Physical Review A—Atomic, Molecular, and Optical Physics}, vol.~74, no.~2, p. 022307, 2006.

\bibitem{shannon1949mathematical}
C.~E. Shannon and W.~Weaver, ``A mathematical model of communication,'' \emph{Urbana, IL: University of Illinois Press}, vol.~11, pp. 11--20, 1949.

\bibitem{hiroka2021quantum}
T.~Hiroka, T.~Morimae, R.~Nishimaki, and T.~Yamakawa, ``Quantum encryption with certified deletion, revisited: Public key, attribute-based, and classical communication,'' in \emph{Advances in Cryptology--ASIACRYPT 2021: 27th International Conference on the Theory and Application of Cryptology and Information Security, Singapore, December 6--10, 2021, Proceedings, Part I 27}.\hskip 1em plus 0.5em minus 0.4em\relax Springer, 2021, pp. 606--636.

\bibitem{hiroka2024certified}
T.~Hiroka, F.~Kitagawa, T.~Morimae, R.~Nishimaki, T.~Pal, and T.~Yamakawa, ``Certified everlasting secure collusion-resistant functional encryption, and more,'' in \emph{Annual International Conference on the Theory and Applications of Cryptographic Techniques}.\hskip 1em plus 0.5em minus 0.4em\relax Springer, 2024, pp. 434--456.

\bibitem{bartusek2023cryptography}
J.~Bartusek and D.~Khurana, ``Cryptography with certified deletion,'' in \emph{Annual International Cryptology Conference}.\hskip 1em plus 0.5em minus 0.4em\relax Springer, 2023, pp. 192--223.

\bibitem{sharifian2021information}
S.~Sharifian and R.~Safavi-Naini, ``Information-theoretic key encapsulation and its application to secure communication,'' in \emph{2021 IEEE International Symposium on Information Theory (ISIT)}.\hskip 1em plus 0.5em minus 0.4em\relax IEEE, 2021, pp. 2393--2398.

\bibitem{cramer1998practical}
R.~Cramer and V.~Shoup, ``A practical public key cryptosystem provably secure against adaptive chosen ciphertext attack,'' in \emph{Advances in Cryptology—CRYPTO'98: 18th Annual International Cryptology Conference Santa Barbara, California, USA August 23--27, 1998 Proceedings 18}.\hskip 1em plus 0.5em minus 0.4em\relax Springer, 1998, pp. 13--25.

\bibitem{Herranz2006KEMDEMNA}
J.~Herranz, D.~Hofheinz, and E.~Kiltz, ``Kem/dem: Necessary and sufficient conditions for secure hybrid encryption,'' \emph{IACR Cryptology ePrint Archive}, 2006.

\bibitem{maurer1993secret}
U.~M. Maurer, ``Secret key agreement by public discussion from common information,'' \emph{IEEE transactions on information theory}, vol.~39, no.~3, pp. 733--742, 1993.

\bibitem{herranz2006kem}
J.~Herranz, D.~Hofheinz, and E.~Kiltz, ``Kem/dem: Necessary and sufficient conditions for secure hybrid encryption,'' \emph{Manuscript in preparation}, 2006.

\bibitem{panja2025hybrid}
S.~Panja, S.~Sharifian, S.~Jiang, and R.~Safavi-Naini, ``Hybrid encryption in correlated randomness model and kem combiners,'' \emph{Theoretical Computer Science}, 2025.

\bibitem{boyer1998tight}
M.~Boyer, G.~Brassard, P.~H{\o}yer, and A.~Tapp, ``Tight bounds on quantum searching,'' \emph{Fortschritte der Physik: Progress of Physics}, vol.~46, no. 4-5, pp. 493--505, 1998.

\bibitem{grover1996fast}
L.~K. Grover, ``A fast quantum mechanical algorithm for database search,'' in \emph{Proceedings of the twenty-eighth annual ACM symposium on Theory of computing}, 1996, pp. 212--219.

\bibitem{11195563}
K.~Dey and R.~Safavi-Naini, ``Secure composition of quantum key distribution and symmetric key encryption,'' in \emph{2025 IEEE International Symposium on Information Theory (ISIT)}, 2025, pp. 1--6.

\end{thebibliography}

\section{Appendix}
\label{Appendix 1}

\begin{theorem}[{\cite[Theorem 3.1]{bartusek2023cryptography}}]
\label{thm:everlsting-security}
Let $\{Z_\lambda(\cdot,\cdot,\cdot)\}_{\lambda\in\mathbb{N}}$ be a quantum operation with three arguments taking as input
a $\lambda$-bit string $\theta$, a bit $b' \in \{0,1\}$, and a $\lambda$-qubit quantum register $\mathsf{A}$. Let $\mathcal{A}$ be a class of adversaries such that for all $\{A_\lambda\}_{\lambda\in\mathbb{N}} \in \mathcal{A}$, for any $\theta \in \{0,1\}^\lambda$, $b' \in \{0,1\}$, and any joint quantum state $\ket{\psi}^{\mathsf{A},\mathsf{C}}$ on a $\lambda$-qubit register $\mathsf{A}$ and an arbitrary auxiliary register $\mathsf{C}$, we have
\begin{equation*}
\begin{split}
\Pr\!\left[A_\lambda\!\left(Z_\lambda(\theta,b',\mathsf{A}),\mathsf{C}\right)=1\right]
-
\Pr\!\left[A_\lambda\!\left(Z_\lambda(0^\lambda,b',\mathsf{A}),\mathsf{C}\right)=1\right] \\
= \text{negl}(\lambda).
\end{split}
\end{equation*}

That is, $Z_\lambda$ is semantically secure against $\mathcal{A}$ with respect to its first input.

For any $\{A_\lambda\}_{\lambda\in\mathbb{N}} \in \mathcal{A}$, define the distribution
$\{Z^{A_\lambda}_\lambda(b)\}_{\lambda\in\mathbb{N},\, b\in\{0,1\}}$ over quantum states by the
following experiment:
\begin{itemize}
    \item Sample $x,\theta \leftarrow \{0,1\}^\lambda$ and initialize $A_\lambda$ with
    \[
    Z_\lambda\!\left(\theta,\, b \oplus \bigoplus_{i:\theta_i=0} x_i,\, \ket{x}_\theta \right).
    \]
    \item Let the output of $A_\lambda$ be parsed as a string $x' \in \{0,1\}^\lambda$ and a residual
    quantum state on register $\mathsf{A}'$.
    \item If $x_i = x'_i$ for all $i$ such that $\theta_i = 1$, output $\mathsf{A}'$; otherwise output a special
    symbol $\bot$.
\end{itemize}

Then,
\[
\mathrm{TD}\!\left(Z^{A_\lambda}_\lambda(0),\, Z^{A_\lambda}_\lambda(1)\right) = \text{negl}(\lambda).
\]
\end{theorem}

We now show that the trace-distance bound established above implies
one-time everlasting certified deletion security as defined in
Definition~\ref{Def: SKE-CD EV-CD}. To this end, we define a quantum
operation $\mathcal{Z}_\lambda(\cdot,\cdot,\cdot)$ with three inputs as
follows:
\begin{enumerate}
    \item Sample $K \leftarrow {\sf D.KeyGen}(1^\lambda)$ and
    $x,\theta \leftarrow \{0,1\}^\lambda$.
    \item Compute
    $m' := m \oplus \bigoplus_{i:\theta_i = 0} x_i$ and prepare the BB84
    state $\ket{x}_\theta$.
    \item Output the joint state
    $(\ket{x}_\theta, {\sf D.Encap}(K,(\theta,b')))$.
\end{enumerate}
Since ${\sf D}$ is an IND-OT secure data encapsulation mechanism (DEM),
it follows that $\mathcal{Z}_\lambda$ semantically hides its first input,
as required by Theorem~3.1 of~\cite{bartusek2023cryptography}.

Let $\rho_\lambda^b$ denote the joint quantum state consisting of the
ciphertext, verification information, and the adversary's internal
registers at the moment when the certificate is verified.
By the correctness of the certified deletion procedure, conditioned on
$\text{\sf D-CD.Vrfy}(cert,vk)=\top$, the resulting residual states
$\rho_\lambda^0$ and $\rho_\lambda^1$ satisfy
\[
\mathrm{TD}(\rho_\lambda^0,\rho_\lambda^1)=\text{negl}(\lambda).
\]
By the operational interpretation of trace distance, this implies that
for any (possibly computationally unbounded) distinguisher
$\mathcal{A}_2$,
\[
\bigl|
\Pr[\mathcal{A}_2(\rho_\lambda^0)=1]
-
\Pr[\mathcal{A}_2(\rho_\lambda^1)=1]
\bigr|
\le \mathrm{TD}(\rho_\lambda^0,\rho_\lambda^1)
= \text{negl}(\lambda).
\]
Therefore, even after receiving the secret key $K$, the second-stage
adversary $\mathcal{A}_2$ cannot distinguish whether $m_0$ or $m_1$ was
encrypted with more than negligible advantage. Consequently,
\[
\text{Adv}_{\text{D-CD},\mathcal{A}}^{\text{ev-cd}}(\lambda)
\le \text{negl}(\lambda),
\]
which establishes one-time everlasting certified deletion security of
{\sf D-CD}.

\end{document}